\documentclass[reqno]{amsart}
\usepackage{graphicx}% Include figure files

\newtheorem{lemma}{Lemma}
\newtheorem{theorem}{Theorem}

\newtheorem{remark}{Remark}

\newcommand{\be}{\begin{eqnarray}}
\newcommand{\ee}{\end{eqnarray}}
\newcommand{\bee}{\begin{eqnarray*}}
\newcommand{\eee}{\end{eqnarray*}}
\newcommand{\R}{{\mathbb R}}

\newcommand{\Z}{{\mathbb Z}}
\newcommand{\C}{{\mathbb C}}

\newcommand{\asy}{{\it O}}

\begin{document}

\title [Double-barrier resonances]{Quantum resonances and time decay for a double-barrier model}

\author {Andrea SACCHETTI}

\address {Department of Physics, Computer Sciences and Mathematics, University of Modena e Reggio Emilia, Modena, Italy}

\email {andrea.sacchetti@unimore.it}

\date {\today}

\thanks {I'm very grateful to Sandro Teta for useful discussions. \ This work is partially 
supported by Gruppo Nazione per la Fisica Matematica (GNFM-INdAM)}

\begin {abstract} {Here we consider the time evolution of a one-dimensional quantum system with a double 
barrier given by a couple of two repulsive Dirac's deltas. \ In such a \emph {pedagogical} model we give, by 
means of the theory of quantum resonances, the explicit 
expression of the dominant terms of $\langle \psi , e^{-it H} \phi \rangle$, where $H$ is the double-barrier 
Hamiltonian operator and where $\psi$ and $\phi$ are two test functions.}

\bigskip

{\it Ams classification (MSC 2010): 35Pxx, 81Uxx, 81U15, 65L15}  

\bigskip


\end{abstract}

\maketitle

\section {Introduction}

The phenomenon of exponential decay associated with quantum resonances is well known since the pioneering works on the Stark effect in an isolated hydrogen atom 
(see the review papers \cite {Ha,Zw}). \ In order to explain such an effect let us consider, in a 
more general context, an Hamiltonian with a discrete eigenvalue ${\mathcal E}_0$ and an associated 
normalized eigenvector $\psi_0$. \ We suppose to weakly perturb such an Hamiltonian and that the new Hamiltonian $H$ has purely absolutely continuous spectrum, 
that is the eigenvalue of the former Hamiltonian disappears into the continuous spectrum. \ Then we physically expect that, after a very short time, one has 
\be
\langle \psi_0 , e^{-itH} \psi_0 \rangle \sim e^{-i t {\mathcal E}} \label {F1}
\ee
where ${\mathcal E}$ is a quantum resonance close to the unperturbed eigenvalu ${\mathcal E}_0$, 
i.e. $\Re {\mathcal E} \sim \Re {\mathcal E}_0$ and  $\Im {\mathcal E} <0$ is such that 
$|\Im {\mathcal E}| \ll 1$. 

The validity of (\ref {F1}) has been proved when the perturbation term is given by a Stark 
potential. \ In such a case Herbst \cite {He} proved that (\ref {F1}) holds true with 
an estimate of the error term (see also \cite {FK} for a extension of such a result to a wider class of models). \ However, we should remark that Simon \cite {S} pointed out that 
the exponentially decreasing behavior is admitted only if the perturbed Hamiltonian $H$ is not bounded from below. \ In fact, 
in the case of Hamiltonian $H$ bounded from below we expect to observe a time decay of the form
\be
\langle \psi_0 , e^{-i t H} \psi_0 \rangle = e^{-i t {\mathcal E}} + b (t) \label {F2}
\ee
where the \emph {remainder} term $b(t)$ is dominant for small and large times, and the 
exponential behavior is dominant for intermediate times. \ On the other hand, dispersive estimates for 
one-dimensional Schr\"odinger operators \cite {KS,W} suggest that for large times the remainder term $b(t)$ is bounded by $c t^{-1/2}$, for some $c>0$, as in the free model 
with no barriers. \ However, this estimate is very raw because it does not take into account the resonances effects.

In this paper we consider a simple one-dimensional model with a double barrier potential with Hamiltonian
\bee
H_\alpha = - \frac {\hbar^2}{2m}\frac {d^2}{dx^2} + \alpha \delta (x+a) + \alpha \delta (x-a) \, . 
\eee
The two barriers are modeled by means of two repulsive Dirac's deltas at $x =\pm a$, for some $a>0$, with 
strength $\alpha \in (0,+\infty ]$. \ This model has been considered by \cite {G} as a \emph {pedagogical} model for the explicit study of quantum barrier resonances. \ When 
$\alpha = + \infty$ the spectrum is purely discrete. \ When $\alpha <+\infty $ then the spectrum of $H_\alpha$ is purely absolutely continuous and the eigenvalues obtained 
for $H_\infty$ disappear into the continuum. \ More precisely, such eigenvalues becomes quantum resonances ${\mathcal E}_{\alpha , n}$ explicitly computed and the time decay of 
$\langle \psi , e^{-i H t}\phi \rangle $, where $\psi$ and $\phi$ are two test functions, has the form (\ref {F2}) where
\be
b(t) = c_\alpha t^{-3/2} + O(t^{-5/2}) \label {F3}
\ee
for large $t$ and for some $c_\alpha >0$ (see Theorem 1); in particular, in the case where the 
two test functions coincide with the unperturbed eigenvector then 
$c_\alpha $ may be explicitly computed (see Theorem 2) and it turns out that 
$c_\alpha \sim \alpha^{-2}$ in agreement with the fact that the asymptotic behavior (\ref {F3}) 
cannot uniformly hold true in a neighborhood of $\alpha =0$. \ We should remark that that our result improves the raw estimate obtained by the decay estimates; in fact, we prove 
that a cancellation effect occurs and that the $t^{-1/2}$ factor, as usually occurs for the free one-dimensional Laplacian problem, is canceled by means of an opposite term 
coming out from the two Dirac's deltas barrier. \ Hence, we can conclude that the effect of the double barrier is twice:

\begin {itemize}

\item [-] the time-decay becomes faster, for $t$ large for any $\alpha >0$;

\item [-] for intermediate times the time-decay is slowed down because of the effect of the quantum resonant states.

\end {itemize}

We finally remark that quantum resonances in one-dimensional double barrier models is a quite interesting 
problem, both for theoretical analysis (see, e.g., \cite {DSSW}) and for possible applications in 
quantum devices (see, e.g., \cite {ChVi}); and thus the explicit and complete analysis in 
a \emph {pedagogical} model would improve the general knowledge of the basic properties.
  
The paper is organized as follows: in Section 2 we introduce the model and we compute the quantum resonant energies; in Section 3 we state our main results and we compare the 
time decay for different values of the parameter $\alpha$; in Section 4 we give the proofs of Theorem 1 and 2; finally, since the calculation of the quantum resonances make 
use of the the Lambert special function we collect its basic properties in an appendix.

\section {Description of the model and quantum resonances}

\subsection {Description of the model}

We consider the resonances problem for a one-dimensional Schr\"odinger equation with two symmetric potential barrier. \ In particular we model the two barrier by means of 
two Dirac's $\delta$ at $x=\pm a$, for some $a>0$. \ The Schr\"odinger operator is formally defined on $L^2 (\R , d x)$ as 
\bee
H_\alpha = - \frac {\hbar^2}{2m}\frac {d^2}{dx^2} + \alpha \delta (x+a) + \alpha \delta (x-a) 
\eee
where $\alpha \in (0,+\infty ]$ denotes the strength of the Dirac's $\delta$. \ Hereafter, for the sake of simplicity let us set $x \to x \sqrt {2m/\hbar^2}$ and 
$a \to a \sqrt {2m/\hbar^2}$; hence the Schr\"odinger operator simply becomes
\bee
H_\alpha = - \frac {d^2}{dx^2} + \alpha \delta (x+a) + \alpha \delta (x-a) 
\eee
where now $a^{-2}$ has the physical dimension of the energy.

When $\alpha <+\infty$ it means that the wavefunction $\psi$ should satisfies to the matching conditions
\be
\psi (x+) = \psi (x-) \ \mbox { and } \ \psi' (x+) = \psi' (x-) + \alpha \psi (x) \ \mbox { at } 
x=\pm a\, , \label {Equa1}
\ee
and $H_\alpha$ has self-adjoint realization on the space of functions $H^2 \left ( \R \setminus \{ \pm a \} \right ) $ 
satisfying the matching conditions (\ref {Equa1}). \ When $\alpha =+\infty$ it means that $H_\infty$ has self-adjoint realization on the space of functions $H^2 \left ( \R \setminus \{ \pm a \} \right ) $ satisfying the Dirichlet conditions
\bee
\psi (\pm a )=0\, .
\eee
In this latter case then the eigenvalue problem
\bee
H_{\infty} \psi = {\mathcal E}_\infty \psi 
\eee
has simple eigenvalues ${\mathcal E}_{\infty , n} = k_n^2 $ where $k_n = \frac {n\pi }{2a}$, $n=1,2,\ldots $, with associated (normalized) eigenvectors
\be
\psi_n (x) = 
\left \{
\begin {array}{ll}
 0 & \mbox { if } \ x< -a \\ 
 \frac {1}{\sqrt {a}} \cos (k_n x ) & \mbox { if } -a < x <+a \\ 
 0 & \mbox { if } +a < x 
 \end {array}
 \right. \, , \ n=1,3,5,\ldots \, , \label {Equa2}
\ee
and
\be
\psi_n (x) = 
\left \{
\begin {array}{ll}
 0 & \mbox { if } \ x< -a \\ 
 \frac {1}{\sqrt {a}} \sin (k_n x ) & \mbox { if } -a < x <+a \\ 
 0 & \mbox { if } +a < x 
 \end {array}
 \right. \, , \ n=2,4,6,\ldots \, . \label {Equa3}
\ee

In the case $\alpha \in (0,+\infty)$ then the eigenvalue problem
\bee
H_\alpha \psi = {\mathcal E}_\alpha \psi 
\eee
has no real eigenvalues, but resonances; where resonances correspond to the complex values of ${\mathcal E}_\alpha$ such that the wavefunction
\bee
\psi (x) = 
\left \{
\begin {array}{ll}
 A e^{i k x} + B e^{-i k x} & \mbox { if } \ x< -a \\ 
 C e^{i k x} + D e^{-i k x} & \mbox { if } -a < x <+a \\ 
 E e^{i k x} + F e^{-i k x} & \mbox { if } +a < x 
 \end {array}
 \right. \, , \ k = \sqrt {{\mathcal E}_\alpha} \, , \ \Im k \ge 0 \, , 
\eee
satisfying the matching condition (\ref {Equa1}), satisfies the \emph {outgoing} condition
\be
A =0 \ \mbox { and } \ F =0\, . \label {Equa4}
\ee

\subsection {Calculation of resonances}

The matching condition (\ref {Equa1}) implies that
\bee
M_1 (-a) \left ( 
\begin {array}{c} 
A \\ B 
\end {array}
\right ) 
= 
M_2 (-a) \left ( 
\begin {array}{c} 
 C \\ D 
\end {array}
\right ) 
\eee
and
\bee 
M_1 (+a) \left ( 
\begin {array}{c} 
 C \\ D 
\end {array}
\right ) 
= 
M_2 (+a) \left ( 
\begin {array}{c} 
 E \\ F 
\end {array}
\right ) 
\eee
where
\bee
M_1 (x) = \left ( 
\begin {array}{cc} 
e^{i k x} & e^{-i k x} \\ 
i k e^{i k x} & - i k e^{-i k x} 
\end {array}
\right ) 
\eee 
and 
\bee
M_2 (x) = \left ( 
\begin {array}{cc} 
e^{i k x} & e^{-i k x} \\ 
(i k -\alpha ) e^{i k x} & - (i k +\alpha ) e^{-i k x} 
\end {array}
\right ) \, .  
\eee
Hence
\bee
\left ( 
\begin {array}{c} 
 E \\ F 
\end {array}
\right ) = M 
\left ( 
\begin {array}{c} 
 A \\ B 
\end {array}
\right ) 
\eee
where
\bee 
M = \left [ M_2(+a) \right ]^{-1} M_1 (+a)  \left [ M_2(-a) \right ]^{-1} M_1 (-a)
\eee
and the resonance condition (\ref {Equa4}) implies that $k$ must satisfies to the following equation
\be
M_{2,2} =0 \, . \label {Equa5}
\ee
A straightforward calculation gives that equation (\ref {Equa5}) takes the form
\be
\frac {1}{4k^2} \left [ e^{4ika } \alpha^2 + 4 k^2 + i 4 k \alpha - \alpha^2 \right ] = 0 \label {Equa6}
\ee
that is, in agreement with equations (284) by \cite {G},
\be
\left ( e^{2ika} \alpha \right ) \pm i \left ( 2 k + i \alpha \right ) =0 \label {Equa7}
\ee
which has two families of complex-valued solutions
\be
k_{1,m} = \frac {i}{2a} \left [ W_m (-a \alpha e^{a \alpha } ) - a\alpha \right ] \label {Equa8}
\ee
and
\be
k_{2,m} = \frac {i}{2a} \left [ W_m (a \alpha e^{a \alpha } ) - a\alpha \right ] \label {Equa9}
\ee
where $W_m(x)$ is the $m$-th branch, $m \in \Z$, of the Lambert special function (in Appendix A we 
recall some basic properties of the Lambert function). \ It turns out that $\Im k_{j,m} <0$ for any $j$ and $m$, but $k_{2,0}=0$, and 
thus equation $H_\alpha \psi = {\mathcal E}_\alpha \psi$ has no eigenvalues for any $\alpha >0$. \ However, we have to remark that for $m<0$ then $\Re k_{j,m} >0$ and 
$\Im k_{j,m} <0$ and then ${\mathcal E}_\alpha = \left ( k_{j,m} \right )^2$ belongs to the \emph {unphysical sheet} with $\Im {\mathcal E}_\alpha <0$ for $m=-1,-2,-3,\ldots $. \ Therefore, we conclude 
that:

\begin {lemma} The spectral problem $H_\alpha \psi ={\mathcal E}_\alpha \psi$ has a family of resonances given by
\be
{\mathcal E}_{\alpha, n} = 
\left \{
\begin {array}{ll} 
k_{1,-(n+1)/2}^2 = \left [ \frac {i}{2a} \left ( W_{-\frac {n+1}{2}}( - a \alpha e^{a \alpha } ) - 
a \alpha  \right ) \right ]^2 & \ \mbox { if } n =1,3,5,\ldots  \\ 
k_{2,-n/2}^2 = \left [ \frac {i}{2a} \left ( W_{-\frac {n}{2} } (  a \alpha e^{a \alpha } )  - 
a \alpha   \right ) \right ]^2 & \ \mbox { if } n =2,4,6,\ldots 
\end {array}
\right.  \, . \label {Equa10}
\ee
\end {lemma}

In Table \ref {tabella1} we report some values of the resonances and we compare these values for different choices of $\alpha$ with the real-valued eigenvalues 
of the problem $H_\infty \psi = {\mathcal E}_\infty \psi $.

\begin{table}
\begin{center}
\begin{tabular}{|l||c|c|c|c|c|} \hline
$n $   & $\alpha =1 $ & $ \alpha =10 $ &   $\alpha =100$   & $\alpha =1000$  &  $\alpha =+\infty$  \\ \hline 
$1$ & $0.57-i2.31$ &  $6.99 -i0.56$    & $9.49 -i1.14 \cdot 10^{-2}$ & $9.83-i1.23 \cdot 10{-4}$  &  $9.87$ \\ \hline
$2$   & $ 13.78 -i 19.35$ &  $29.26 -i3.86$  & $37.95 -i9.1 \cdot 10^{-2}$ & $39.32 - i 9.84 \cdot 10^{-4}$  &  $39.48$ \\ \hline
$3$  & $49.60 -i41.74$ &  $69.10-i10.55$   & $85.41 -i3.05 \cdot 10^{-1}$ & $88.47 -i 3.32 \cdot 10^{-3}$  &  $88.83$ \\ \hline
$4$  & $106.16 -i66.64$ & $127.91 - i 19.99$   & $151.89 -i7.13 \cdot 10^{-1}$ & $157.28 - i 7.87 \cdot 10^{-3}$  &  $157.91$ \\ \hline
\end{tabular}
\caption{\small Table of values of the resonances ${\mathcal E}_{\alpha ,n}$, $n=1,2,3,4$, given by (\ref {Equa10}) for $a= \frac 12$ and for $\alpha =1,\ 10,\ 100$ and 
$\alpha =1000$; for $\alpha =+\infty$ these values correspond to the real-valued eigenvalues of the problem $H_\infty \psi = {\mathcal E}_\infty \psi $.}
\label{tabella1}
\end{center}
\end {table}

\begin {remark}
Let $a>0$ be fixed then, from (\ref {Equa24}), it follows that for $n$ fixed and $\alpha$ large enough the asymptotic behavior of the resonances is given by
\bee
{\mathcal E}_{\alpha,n} &=& \left [ \frac {n\pi}{2a} \left ( 1 - \frac {1}{a\alpha} + \frac {1}{(a\alpha)^2} \right ) - i \frac {1}{2a} 
\left ( \left ( \frac {\ln (a\alpha )}{a \alpha } \right )^2 + \frac {n^2 \pi^2 }{2(a \alpha )^2} \right ) + \asy \left ( \alpha^{-3} \right ) \right ]^2 \\ 
&\sim & \left ( \frac {n \pi }{2 a } \right )^2 - i \frac {n \pi}{2 a^2} \left ( \frac {\ln (a \alpha )}{a \alpha } \right )^2 
\eee
\end {remark}

\subsection {Resolvent operator}

The explicit form of the resolvent of $H_\alpha$, $\alpha \in (0,+\infty)$ is given by \cite {Al} 
\bee
\left ( \left [ H_{\alpha } - k^2 \right ]^{-1} \phi \right )(x) =
\int_{\R} K_\alpha (x,y;k) \phi (y) d y,
\quad \phi\in L^2(\R) , \ \Im k \ge 0 \, ,
\eee
where the integral kernel $K_\alpha$ is given by
\be \label{Equa11}
K_\alpha (x,y;k) = K_0 (x,y;k) + \sum_{j=1}^4 K_j (x,y;k)
\ee
with
\bee
K_0 (x,y;k) = \frac {i}{2k}\,  e^{ik|x-y|}
\eee
and
\bee 
K_j (x,y;k) = g(k) L_j (x,y;k) \, , 
\eee
where
\be
g(k):= - \left [ 2k((2k+i \alpha )^2 + \alpha^2 e^{i4ka}) \right ]^{-1} \, ,  \label {Equa12} 
\ee
and
\bee
L_1 (x,y;k) = - \alpha(2k+i\alpha)\,  e^{ik |x+a|} e^{ik|y+a|}, &
L_4 (x,y;k)  = L_1 (-x,-y;k) \nonumber \\
L_2 (x,y;k) = i\alpha^2\, e^{2ika}\,  \, e^{ik |x+a|} e^{ik|y-a|}, &
L_3 (x,y;k)  = L_\alpha^2 (-x,-y;k) . 
\eee
Resonances can be defined as the complex poles in the \emph {unphysical sheet} $\Im {\mathcal E}_\alpha <0$ of the kernel of 
the resolvent, too; that is the pole of the function $g(k)$ in agreement with (\ref {Equa7}).

\section {Time decay - main results}

Let $\phi$ and $\psi$ two well localized wave-functions, we are going to estimate the time decay of the term
\be
\langle \psi , e^{-it H_\alpha} \phi \rangle \label {Equa12Bis}
\ee

\begin {theorem}
Let us assume that $\phi$ and $\psi$ have compact support. \ Then 
we have that
\be
\langle \psi , e^{-it H_\alpha} \phi \rangle = c_\alpha t^{-3/2} + \sum_{n=1}^\infty \beta_n c_n e^{-i {\mathcal E}_{\alpha ,n} t} + O (t^{-5/2}) \label {Equa13}
\ee
for some constants $c_\alpha$ and $c_n$ and where 
\be
\beta_n = 
\left \{ 
\begin {array}{ll}
1 & \ \mbox { if } \ \left | \Im \sqrt {{\mathcal E}_{\alpha ,n}} \right | <  \left | \Re \sqrt {{\mathcal E}_{n,\alpha }} \right | \\ 
\ \\ 
\frac 12 & \ \mbox { if } \ \left | \Im \sqrt {{\mathcal E}_{\alpha ,n}} \right | =  \left | \Re \sqrt {{\mathcal E}_{\alpha ,n}} \right | \\ 
\ \\ 
0 & \ \mbox { if } \ \left | \Im \sqrt {{\mathcal E}_{\alpha ,n}} \right | >  \left | \Re \sqrt {{\mathcal E}_{\alpha ,n}} \right | 
\end {array}
\right. \, . \label {Equa13Bis} 
\ee
\end {theorem}

\begin {remark}
We may remark that in the case $\alpha =0$, that is when there are no barriers, then $\langle \psi , e^{-it H_0} \phi \rangle 
\sim t^{-1/2}$ and an apparent contradiction appears. \ The point is that the asymptotic expansion (\ref {Equa13}) is not 
uniform as $\alpha$ goes to zero. \ In fact, in an explicit model, see Theorem 2, it results that $c_\alpha \to \infty$ as $\alpha \to 0$.
\end {remark}

We consider now, in particular, the asymptotic behavior of (\ref {Equa12Bis}) when the test vectors $\phi$ and $\psi$ coincide with one of 
the \emph {resonant} states, e.g. with $\psi_1$.

\begin {theorem}
Let $\psi =\phi$ coinciding with the eigenvector $\psi_1$, let $\ell (k)$ be the function defined as 
\be
\ell (k) = 2\pi \sqrt {a}  \frac {e^{2kai}+1}{\pi^2-4k^2a^2}\, , \label {Equa22}
\ee
and let $k_{1,-m}$ be the resonances defined by (\ref {Equa8}). \ Then
\bee
\langle \psi_1 , e^{-it H_\alpha} \psi_1 \rangle = c_\alpha t^{-3/2} + \sum_{m=1}^\infty \beta_m c_m e^{-i k_{1,-m}^2t} + O(t^{-5/2})
\eee
where 
\bee
c_\alpha = - \frac {2^{3/2}(1+i)a}{\pi^{5/2} \alpha^2}\, , \ \ c_m = - 
\frac {\alpha \ell (k_{1,-m})^2  }{1+\alpha a \left ( 1 + \frac {2k_{1,-m}}{i \alpha } 
\right )} 
\eee
and where $\beta_m$ is defined by (\ref {Equa13Bis}).
\end {theorem}

\begin {remark}
Recall that
\bee
\ell (k_{1,-m}) = \frac {4i\pi \sqrt {a} k_{1,-m}}{\alpha (\pi^2 - 4 k_{1,-m}^2 a^2)}
\eee
Hence
\bee
\ell (k_{1,-m}) \sim O(\alpha^{-1}) \ \mbox { if } \ m\not= 1
\eee
as $\alpha \to +\infty$. \ For $m=1$, from the asymptotic behavior of $k_{1,-m}$ it follows that 
\bee
\pi^2 - 4 k_{1,-1}^2 a^2 \sim  \pi^2 - 4 a^2 \left [ \frac {\pi^2}{4a^2} \left ( 1 - \frac {2}{a\alpha} \right ) \right ] = 
\frac {2\pi^2}{a\alpha}
\eee
and then
\bee
\ell (k_{1,-1}) \sim   i \sqrt {a} \, ,\ \mbox { as } \ \alpha \to + \infty \, . 
\eee
Hence, as $\alpha$ goes to infinity it follows that the dominant term of $\langle \psi_1 , e^{-it H_\alpha} \psi_1 \rangle$ is given by
\bee
\langle \psi_1 , e^{-it H_\alpha} \psi_1 \rangle =  e^{-i \frac {\pi}{2a} t} + O(\alpha^{-1}) 
\eee
in agreement with the fact that $\langle \psi_1 ,e^{-iH_\infty t} \psi_1 \rangle = e^{-i {\mathcal E}_{\alpha ,1} t}$.
\end {remark}

\begin {remark} Let us compare, in the limit of large $\alpha$, the two dominant terms of $\langle \psi_1 , e^{-it H_\alpha} \psi_1 \rangle$; that is the power term 
\bee
\frac {d_1}{\alpha^2 t^{3/2}} \, ,\ \mbox { where } d_1 = \left | \frac {2^{3/2}(1+i)a}{\pi^{5/2} } \right | = \frac {4a}{\pi^{5/2}}\, , 
\eee
and the exponential term
\bee
d_3 e^{\Im {\mathcal E}_{\alpha ,1} t} = d_3 e^{-d_2 t \left ( \frac {\ln (a\alpha )}{a \alpha } \right )^2 } \, ,  \mbox { where } \ d_2 = 
\frac {\pi}{2a^2} \ \mbox { and } \ d_3= |c_1| = 
\left | \frac {\alpha \ell (k_{1,-m})^2 }{1+\alpha a \left ( 1 + \frac {2k_{1,-m}}{i \alpha } 
\right )} \right | \sim 1 \, . 
\eee
In order to understand when the power behavior dominates and when the exponential behavior dominates we have to solve the inequality
\bee
\frac {d_1}{\alpha^2 t^{3/2}} < d_3 e^{-d_2 t \left ( \frac {\ln (a\alpha )}{a \alpha } \right )^2 } \, . 
\eee
A straightforward calculation gives that this inequality is satisfied for any $t\in [t_1 , t_2]$, where $0<t_1 <t_2$ are given by 
\bee
t_{1} = -\frac {3}{\pi} W_{0} (z) \frac {a^4 \alpha^2}{ \ln^2 (a\alpha)} \ \mbox { and } \ t_2 = -\frac {3}{\pi} W_{0} (-z) 
\frac {a^4 \alpha^2}{ \ln^2 (a\alpha)} 
\eee
where
\bee
z=  -\frac {2^{4/3}}{3 \pi^{2/3}} \frac {\ln^2 (a\alpha)} { (a\alpha)^{10/3}  } \, .
\eee
This interval is not empty provided that the argument $z$ of the Lambert function is between $(-1/e ,0)$; that is 
\bee
\frac {2^{4/3}}{3 \pi^{2/3}} \frac {\ln^2 (a\alpha)} { (a\alpha)^{10/3}  }  < \frac 1e
\eee
which holds true for 
\bee
a \alpha > \mbox {\rm exp} \left [ - \frac 35 W_0 \left ( \frac 56 (2\pi )^{1/3} \sqrt {2/e}  \right ) \right ] \approx 0.635 \, . 
\eee
\end {remark}

\begin {remark}
In Fig. \ref {decadimento} we plot the graph of $\langle \psi_1 , e^{-it H_\alpha} \psi_1 \rangle$ for different values of $\alpha$: 
$\alpha =0,\ 0.1, \ 1$ and $10$. \ We numerically 
compute the integral $\langle \psi_1 , e^{-itH_\alpha }  \psi_1 \rangle$ and we observe, in fully agreement with 
Theorem 2, that when $\alpha >0$ then the time decay, for large $t$, is faster of the one obtained for $\alpha =0$; on the other side we 
observe that for intermediate times the time decay is delayed as effect of the resonant states.
\end {remark}

\begin{figure} [h]
\begin{center}
\includegraphics[height=8cm,width=8cm]{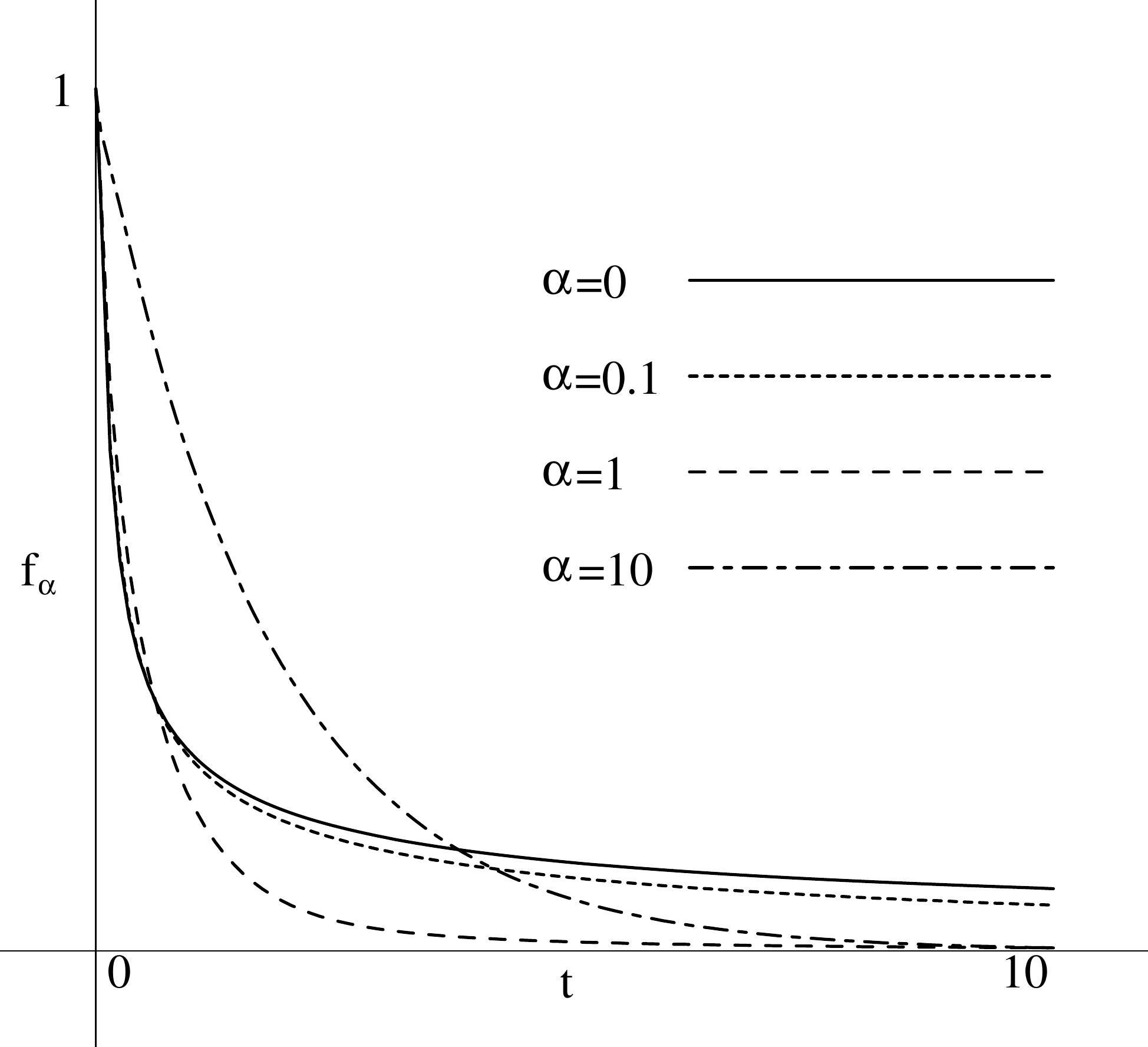}
\caption{Plot of the graph of $\langle  \psi_1 , e^{-itH_\alpha }  \psi_1 \rangle$, where $\psi_1$ is the eigenvector of $H_\infty$, 
for different values of $\alpha$.
\label {decadimento}}
\end{center}
\end{figure}

\section {Proofs of Theorems 1 and 2}

First of all we prepare the ground for the proofs. \ Then, we prove Theorem 2, at first; the proof of Theorem 1 will follow as a natural 
extension of the proof of Theorem 2.

\subsection {General setting}

Let us denote
\bee
f_\alpha (t) := \langle \psi , e^{-it H_\alpha} \phi \rangle
\eee
where $\psi $ and $\psi$ are two test functions. \ By means of standard arguments we can exchange the order of integration obtaining that
\bee
f_\alpha (t) &:=& \lim_{\epsilon \to 0^+} \langle \psi , e^{-(it +\epsilon) H_\alpha} \phi \rangle \\ 
&=& \lim_{\epsilon \to 0^+} \frac {1}{\pi i} \int_\R \int_\R \overline {\psi (x)} \phi (y) \int_\R k K_\alpha (x,y,k) e^{-k^2 (it+\epsilon )} \, dk \, dy \, dx \\ 
&=& \lim_{\epsilon \to 0^+} \frac {1}{\pi i} \int_\R \left [ \int_\R \int_\R \overline {\psi (x)} \phi (y) K_\alpha (x,y,k)  \, dy \, dx \right ] 
k e^{-k^2 (it+\epsilon )} \, dk
\eee
Hence, it turns out that 
\bee
f_\alpha (t) = f_0(t) + f_r (t), \ f_r (t) := \sum_{j=1}^4 f_j (t) 
\eee
where $f_0 (t)$ is the time evolution term associated to the free Laplacian operator given by 
\be
f_0 (t) =  \lim_{\epsilon \to 0^+} \frac {1}{\pi i} \int_\R G_0(k) 
k e^{-k^2 (it+\epsilon )} \, dk \label {Equa14}
\ee
with kernel 
\be
G_0 (k):= \int_\R \int_\R \overline {\psi (x)} \phi (y) K_0 (x,y,k)  \, dy \, dx \, . \label {Equa15}
\ee
The term $f_r$ is due to the effect of the Dirac's delta barriers:
\bee
f_j (t) &=&    \lim_{\epsilon \to 0^+} \frac {1}{\pi i} \int_\R \left [ \int_\R \int_\R \overline {\psi (x)} \phi (y) K_j (x,y,k)  \, dy \, dx \right ] 
k e^{-k^2 (it+\epsilon )} \, dk  \\ 
&=& \lim_{\epsilon \to 0^+} \frac {1}{\pi i} \int_{\R} G_j (k) k g(k) e^{-(it+\epsilon ) k^2} dk                            
\eee
where
\bee
G_j (k) = \int_{\R} \int_{\R} \overline {\psi (x)} \phi (y) L_j (x,y,k) dx dy 
\eee
are such that 
\bee
G_1 (k) &=& - \alpha (2k+i\alpha) \ell_1 (k) m_1(k) \\ 
G_2 (k) &=& i \alpha^2 e^{2ika} \ell_1 (k) m_2 (k) \\ 
G_3 (k) &=& i \alpha^2 e^{2ika} \ell_2 (k) m_1 (k) \\ 
G_4 (k) &=& - \alpha (2k+i\alpha) \ell_2 (k) m_2(k)
\eee
where
\be
\ell_1 (k) = \int_{\R} e^{ik|x+a|} \overline {\psi (x)} dx \, , \ \ell_2 (k) = \int_{\R} e^{ik|x-a|} 
\overline {\psi (x)} dx \label {Equa16} \\ 
m_1 (k) = \int_{\R} e^{ik|y+a|} \phi (y) dy \, , \ m_2 (k) = \int_{\R} e^{ik|y-a|} \phi (y) dy \label {Equa17}
\ee

In conclusion, we have proved that 
\begin {lemma}
Let $\phi$ and $\psi$ two given well localized wave-function, e.g. $\phi ,\psi \in L^2\cap L^1$, then the term
\bee
f_\alpha (t) := \langle \psi , e^{-iH_\alpha t} \phi \rangle 
\eee
is given by the sum of two terms $f_0 (t)$ and $f_r (t)$ where $f_0 (t)$ is the time evolution term   
given by (\ref {Equa14}), associated to the free Laplacian, and where the other term $f_r (t)$ is due to the effect of the 
two Dirac's delta barriers 
\bee
f_r (t) = \lim_{\epsilon \to 0^+} \frac {1}{\pi i} \int_{\R} G_r (k) k g(k) e^{-(it +\epsilon) k^2} d k 
\eee
where
\be
G_r (k) := \sum_{j=1}^4 G_j (k) = -\alpha (2k+i\alpha ) \left [\ell_1 m_1 + 
\ell_2 m_2 \right ] + i \alpha^2 e^{2ika} \left [\ell_1 m_2 + \ell_2 m_1 \right ]
\label {Equa18}
\ee
and where $g(k)$ is the function defined by (\ref {Equa12}).
\end {lemma}

\begin {remark} If $\psi  = \phi$ and it is a real valued function then $\ell_j = m_j$ and 
\be
G_r (k) = -\alpha (2k+i\alpha ) \left [\ell_1^2 + \ell_2^2 \right ] + 2i \alpha^2 e^{2ika} \ell_1 \ell_2 \label {Equa19}
\ee
Furthermore, if
\begin {itemize}

\item [-] $\psi$ is an even function, i.e. $\psi (-x) = \psi (x)$, then $\ell:= \ell_1 = \ell_2$ and 
\be
G_r (k) = 2 \ell^2 (k) \left [ -\alpha (2k+i \alpha ) + i \alpha^2 e^{2ika} \right ] \label {Equa20}
\ee

\item [-] $\psi$ is an odd function, i.e. $\psi (-x) = -\psi (x)$, then $\ell:= \ell_1 = -\ell_2$ and 
\be
G_r (k) &=& 2 \ell^2 (k) \left [ -\alpha (2k+i \alpha ) - i \alpha^2 e^{2ika} \right ] \label {Equa21}
\ee

\end {itemize}

\end {remark}

\subsection {Proof of Theorem 2}

Assume that $\psi $ and $\phi$ coincide with some eigenvector of $H_\infty$. \ In particular, for the sake of definiteness 
we assume that $\phi =\psi =\psi_1$, where $\psi_1$ is the eigenvector (\ref {Equa2}) of $H_\infty$ associated with the ground state. \  
Then a straightforward calculation gives that $\ell (k)$ is given by (\ref {Equa22}) and that 
\bee
G_0 (k) = \frac {4ia}{k} \frac {i a k (4 k^2 a^2 - \pi^2 ) + \pi^2 (1+e^{2ika} )}{(4 k^2 a^2 - \pi^2 )^2}\, . 
\eee
Furthermore, from (\ref {Equa20}) it follows that 
\bee
G_r (k) k g(k) = \frac {\alpha \ell^2 (k)}{(2k+i\alpha)+ i \alpha e^{2ika}}
\eee
is a meromorphic function with simple poles at $k= k_{1,m}$, $m=-1,-2,\ldots $, and thus
\be
f_r (t):=  \frac {\alpha}{\pi i} \int_{\R} d k 
\frac {\ell^2 (k)}{(2k+i\alpha)+ i \alpha e^{2ika}} e^{-i k^2t} \label {Equa23}
\ee

\begin {remark}
In general, in the case $\phi =\psi = \psi_n$, where $n=1,3,5,\ldots $, from (\ref {Equa20}) it follows that the argument of the integral 
in (\ref {Equa23}) is a meromorphic function with simple poles at $k=k_{1,m}$, $m=-1,-2,-3,\ldots $; in the case $\phi =\psi = \psi_n$, 
where $n=2,4,6,\ldots $, from (\ref {Equa21}) it follows that the argument of the integral in (\ref {Equa23}) is
a meromorphic function with simple poles at $k=k_{2,m}$, $m=-1,-2,-3,\ldots $.
\end {remark}

We prove the Theorem in two steps. \ In the first step we compute the asymptotic behavior of $f_0 (t)$, while in the second step 
we compute the asymptotic behavior of $f_r (t)$.

\subsubsection {Asymptotic behavior of $f_0 (t)$ for large $t$}

\begin {lemma} \label {lemma3}
Let $\psi = \phi$ coinciding with the eigenvector $\psi_1$. \ Then 
\bee
f_0 (t) = \frac {8a}{\pi^{5/2} t^{1/2} } e^{-i\pi/4} + \frac {2^{3/2} a^3 (\pi^2 -8)(1+i)}{t^{3/2} \pi^{9/2}} + O(t^{-5/2}) \ 
\mbox { as } \ t \to + \infty \, . 
\eee
\end {lemma}

\begin {proof} We have that 
\bee
f_0 (t) := \frac {1}{\pi i} \int_{\R} G_0 (k) k e^{-i k^2t} d k = \frac {4a}{\pi } \int_{\R} h (k) e^{-i k^2t} d k 
\eee
where we set
\bee 
h(k):= \frac {i a k (4 k^2 a^2 - \pi^2 ) + \pi^2 (1+e^{2ika} )}{(4 k^2 a^2 - \pi^2 )^2} 
\eee
and where we observe that at the roots $k =\pm k_1 =\pm \frac {\pi}{2a}$ of the denominator of the integrand function then
\bee
\lim_{k \to \pm k_1} h(k) = \frac {1}{8} \frac {\pi \pm 3 i}{\pi}
\eee
and that the integral converges because for large values of $k$ the integrand function behaves as the integrable (not in absolute value) 
function $\frac {e^{-i k^2t}}{k}$.

Now, let $R>0$ be fixed and let 
\bee
\gamma = \gamma_1 \cup \gamma_2 \cup \gamma_3 \cup \gamma_4
\eee
where $\gamma_1 =[-R,+R]$, and 
\bee
\gamma_2 &=& \{ k \in \C \ :\ k = R e^{i\theta }, \ \theta \in [-\pi/4 ,0] \} \\ 
\gamma_3 &=& \{ k \in \C \ :\ k = - r e^{i\theta }, \ r \theta \in [-R ,R] \} \\ 
\gamma_4 &=& \{ k \in \C \ :\ k = R e^{i\theta }, \ \theta \in [\pi ,3\pi/4] \} 
\eee
with $\gamma_2$ and $\gamma_4$ clock-wise oriented (see Fig. \ref {Fig1}). \ From the Cauchy theorem it follows that 
\bee
\oint_\gamma   h(k) e^{-i k^2t} d k =0 
\eee
and then 
\bee
f_0 (t) := \lim_{R\to + \infty}  \frac {4a}{\pi } \int_{\gamma_1} h(k) e^{-i k^2t} d k 
 = - \sum_{j=2}^4 \lim_{R\to + \infty} \frac {4a}{\pi } \int_{\gamma_j} h(k) e^{-i k^2t} d k 
\eee
\begin{figure} [h]
\begin{center}
\includegraphics[height=6cm,width=10cm]{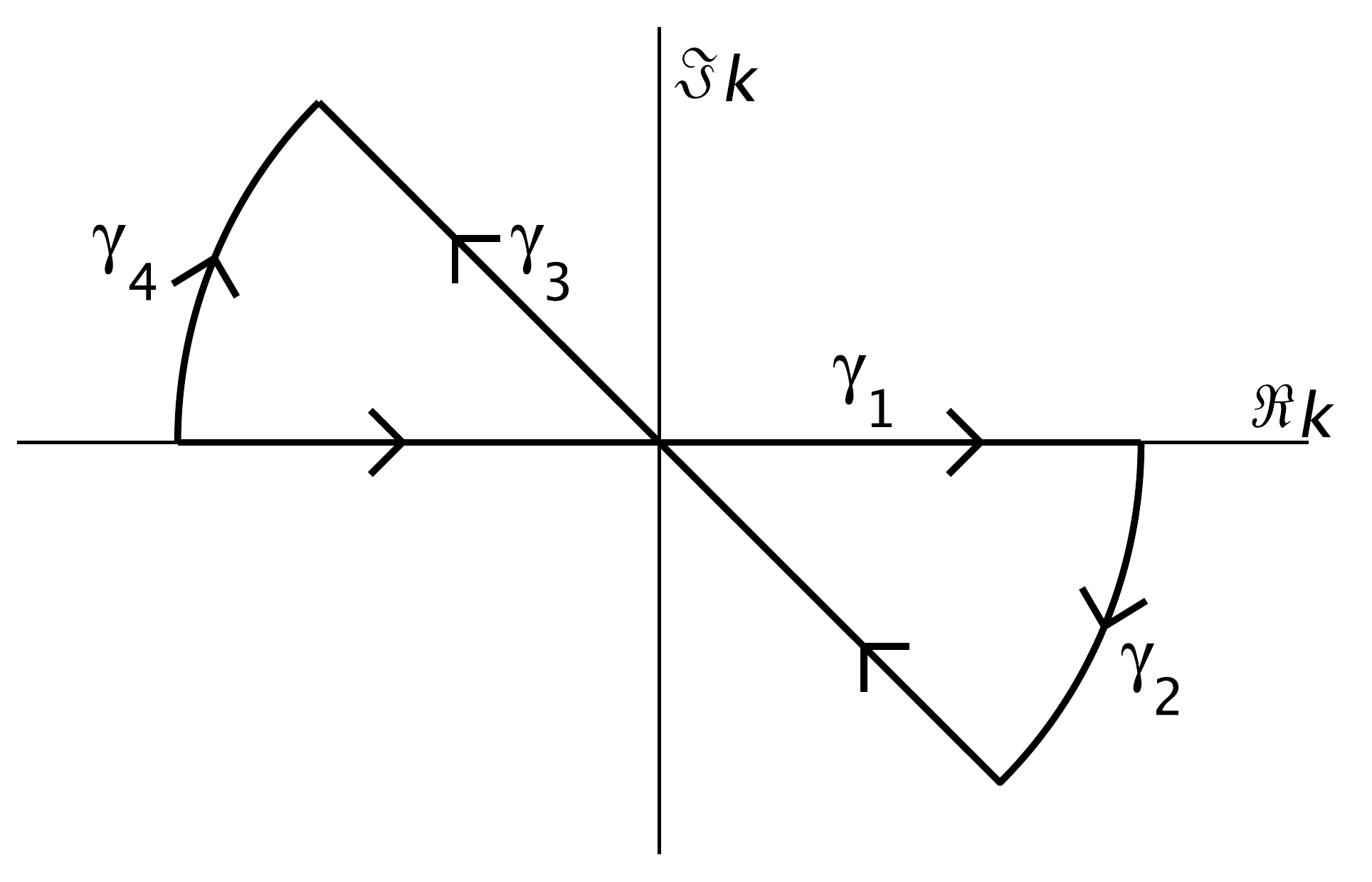}
\caption{\label {Fig1}}
\end{center}
\end{figure}

First of all we prove that 
\begin {lemma} \label {lemma2}
\bee
\lim_{R\to + \infty}   \frac {4a}{\pi } \int_{\gamma_j} h(k) e^{-i k^2t} d k = 0 \ \mbox { for } \ j=2,4 \, . 
\eee
\end {lemma}

\begin {proof}
Indeed, for $j=2$ (for instance), then $k = R e^{i\theta }$, $\theta \in [-\pi/4 ,0]$, and thus (remember that $\gamma_2$ is clock-wise 
oriented)
\bee
\frac {4a}{\pi } \int_{\gamma_2} h(k) e^{-i k^2t} d k = -\frac {4a}{\pi } \int_{-\pi/4}^0 h(Re^{i\theta }) e^{-i R^2 t e^{2i\theta}} R d 
\theta  
\eee
Hence
\bee
\left | \frac {4a}{\pi } \int_{\gamma_2} h(k) e^{-i k^2t} d k \right | 
&\le & C \int _{-\pi/4}^0 \left [ 1+ \frac {e^{-2Ra\sin \theta }}{R^3} \right ] e^{2 R^2 t \sin \theta \cos \theta } d \theta \\
&\le & C \int _{-\pi/4}^0 \left [ 1+ \frac {e^{-2Ra\sin \theta }}{R^3} \right ] e^{\sqrt {2} R^2 t \sin \theta  } d \theta \to 0 \ 
\mbox { as } \ R \to + \infty \, .
\eee
since $\cos \theta \ge \frac {\sqrt {2}}{2}$ for $\theta \in [-\pi /4,0]$.
\end {proof}

Therefore, we can conclude that
\bee
f_0 (t) &=& \frac {4a}{\pi } e^{-i\pi /4} \int_{-\infty}^{+\infty} h(ke^{-i\pi /4}) e^{-i (ke^{-i\pi /4})^2t} d k \\ 
 &=& \frac {4a}{\pi } e^{-i\pi /4} \int_{-\infty}^{+\infty} h(ke^{-i\pi /4}) e^{- k^2 t} d k \\
&=& \frac {8a}{\pi^{5/2} t^{1/2} } e^{-i\pi/4} + \frac {2^{3/2}ia^3 (\pi^2 -8)(1+i)}{t^{3/2} \pi^{9/2}}  + O(t^{-5/2}) \ \mbox { as } \ t 
\to + \infty 
\eee
by Watson's Lemma (see pg. 402 by \cite {SFS}). \ Lemma \ref {lemma3} is thus proved.
\end {proof}

\begin {remark} \label {Nota1} It is well known (see Lemma 7.10 at page 169 by \cite {T}) that
\bee
\left [ e^{-i t H_0 } \phi \right ] (x) &=& \frac {1}{\sqrt {2 i t}} e^{ix^2 /4t} 
\widehat {\left ( e^{iy^2/4t} \phi (y) \right )} \left ( \frac {x}{2t} \right ) \\
&\sim &  \frac {1}{\sqrt {2 i t}} e^{ix^2 /4t} \hat \phi (x/2t) 
\eee
in norm as $t$ goes to infinity. \ Then, when $\psi$ and $\phi$ coincide with the eigenvector $\psi_1$ an explicit calculation gives that
\bee
\hat \psi_1 (\omega ) =  \sqrt {2\pi a}  \frac {1+e^{i 2\omega a } }{e^{i\omega a}(-4\omega^2a^2+\pi^2)} \, . 
\eee 
Hence
\bee
\langle \psi_1 ,  e^{-i t H_0 } \psi_1 \rangle &=& \int_{-a}^{+a} \sqrt {\frac {\pi}{it}} \cos (\pi x /2a) \frac {\left ( 1+e^ {i x a /t} \right ) e^{i x^2 /4t}}{ e^{ixa/2t} \left (\pi^2 -a^2 x^2/t^2 \right )  } dx \\ 
&=& \sqrt {\frac {\pi}{it}} \int_{-a}^{+a}  \cos (\pi x /2a) \frac {\left ( 1+e^ {i x a /t} \right ) e^{i x^2 /4t}}{ e^{ixa/2t} \left (\pi^2 -a^2 x^2/t^2 \right )  }dx \\ 
&\sim &\sqrt {\frac {\pi}{it}} \int_{-a}^{+a}  \cos (\pi x /2a) \frac {2}{ \pi^2 }dx = \frac {8 a\sqrt {-\pi i}  }{t^{1/2}{\pi}^3 } 
\eee
in agreement with Lemma \ref {lemma3}.
\end {remark}

\subsubsection {Asymptotic behavior of $f_r (t)$ for large $t$}

\begin {lemma} 
Let $\psi =\phi$ coinciding with the eigenvector $\psi_1$. \ Then
\bee
&& f_r (t) = - \frac {8a}{\pi^{5/2} \sqrt {t}} e^{-i\pi /4}  - \frac {2^{3/2}(1+i)a^3\left ( \pi^2  - 
8 \right )}{\pi^{9/2} t^{3/2}} - \frac {2^{3/2}(1+i)a}{\pi^{5/2} t^{3/2}\alpha^2} + \\ 
&& \ \ - \sum_{m=1}^\infty \frac {\alpha \beta_m \ell (k_{1,-m})^2}{1+\alpha a \left ( 1 + \frac {2k_{1,-m}}{i \alpha } 
\right )} e^{-i k_{1,-m}^2t} + O(t^{-5/2})
\eee
where $\beta_m$ is defined by (\ref {Equa13Bis}).
\end {lemma}

\begin {proof}
We recall (\ref {Equa23}), where the function $\ell (k)$ has been defined by (\ref {Equa22}) and where we observe that at the roots 
$k =\pm k_1 = \pm \frac {\pi}{2a}$ of the denominator of $\ell (k)$ then
\bee
\lim_{k \to \pm k_1} \ell (k) = \frac {1}{8} \frac {\pi \pm 3 i}{\pi} \, . 
\eee
Furthermore, the integral converges because for large values of $k$ the integrand function behaves as the integrable (in absolute value) 
function $\frac {1}{k^5}$. \ The integrand function has poles coinciding with the roots of $(2k+i\alpha ) + i \alpha e^{2ika}$, that is at
\bee
k_{1,j} = \frac {i}{2a} \left [ W_j(-a\alpha e^{a\alpha } ) - a \alpha \right ] \, , \ j \in \Z \, . 
\eee

Now, let $R>0$ be fixed and let (see Fig.  \ref {Fig2}) 
\bee
\gamma = \gamma_1 \cup \gamma_2 \cup \gamma_3 \cup \gamma_4
\eee
\begin{figure} [h]
\begin{center}
\includegraphics[height=6cm,width=10cm]{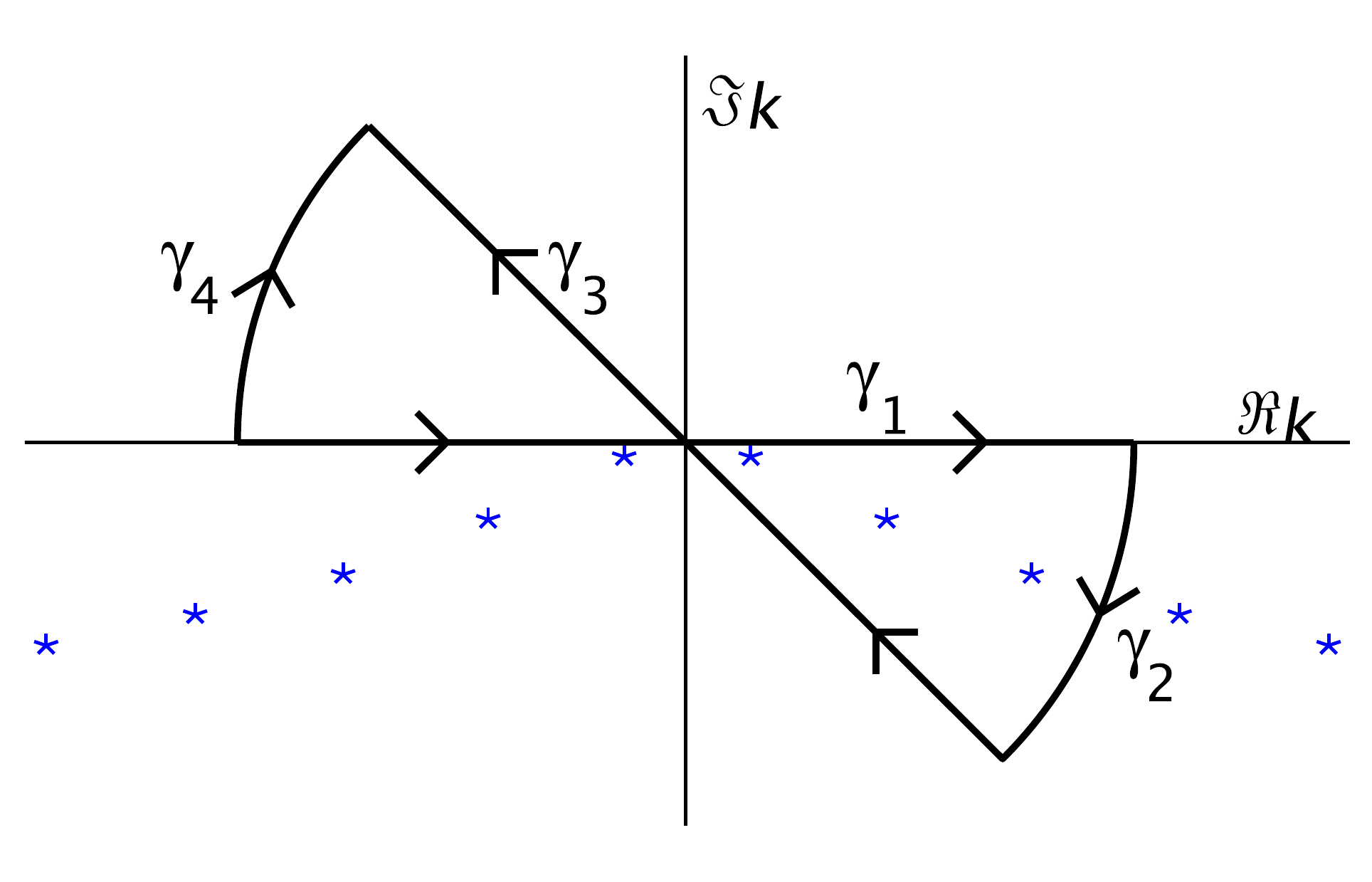}
\caption{Asterisks denote the poles $k_{1,j}$, $j\in \Z$, of the integrand function in (\ref {Equa23}). \label {Fig2}}
\end{center}
\end{figure}
as in the proof of Lemma \ref {lemma3}, where $R>0$ is such that $k_{1,j} \notin \gamma$. \ Since the integrals along $\gamma_2$ and $\gamma_4$ go to zero as $R$ 
goes to infinity, as in Lemma \ref {lemma2}, and from the Cauchy theorem it follows that 
\bee
f_r (t) = I + II 
\eee
where 
\bee
I &=&  -\frac {\alpha}{\pi i} \int_{\gamma_3} \frac {[\ell (k)]^2}{(2k+i\alpha ) + i \alpha e^{2ika}} 
e^{- ik^2t} d k \\ 
&=& \frac {\alpha}{\pi i} \int_{\R} \frac {\ell (ke^{-i\pi/4})^2}{(2ke^{-i\pi/4}+i\alpha ) + i 
\alpha e^{2ike^{-i\pi/4}a}} e^{- k^2t} e^{-i\pi/4} d k \\ 
%&=&  - \frac {8a}{\pi^{5/2} \sqrt {t}} e^{-i\pi /4} - \frac {2^{3/2}(1+i)a\left ( a^2 \pi^2 \alpha^2 - 
%8 a^2 \alpha^2 +\pi^2 \right )}{\pi^{9/2} t^{3/2}\alpha^2} + O(t^{-5/2}) \\ 
&=&  - \frac {8a}{\pi^{5/2} \sqrt {t}} e^{-i\pi /4}  - \frac {2^{3/2}(1+i)a^3\left ( \pi^2  - 
8 \right )}{\pi^{9/2} t^{3/2}} - \frac {2^{3/2}(1+i)a}{\pi^{5/2} t^{3/2}\alpha^2}+ O(t^{-5/2})
\eee
by means of the Watson's Lemma, and where the term II picks up the contribution due to the residue at the poles $k_{1,-m}$, 
$m=1,2,3,\ldots $, such that $|\Im k_{1,-m}| \le |\Re k_{1,-m}| $:
\bee
II &=& -2{\alpha}  \sum_{m =1}^\infty \beta_m \mbox {Res} \left [ \frac {\ell (k)^2}{(2k+i\alpha ) + i \alpha e^{2ika}} e^{-i k^2t} \right ]_{k=k_{1,-m}} \\ 
&=& -2\alpha \sum_{m=1}^\infty \frac {\beta_m \ell (k_{1,-m})^2}{2+2\alpha a \left ( 1 + \frac {2k_{1,-m}}{i \alpha } 
\right )} e^{-i k_{1,-m}^2t} 
\eee
where $\beta_m =1$ if the pole $k_{1,-m}$ is inside the complex set enclosed by $\gamma$, and where $\beta_m =\frac 12$ 
if the pole $k_{1,-m}$ belongs to the border $\gamma$, otherwise $\beta_m=0$. \ Lemma 5 is so proved.
\end {proof}

\subsection {Proof of Theorem 1}

In fact, the proof of Theorem 1 is nothing but the extension of the proof of Theorem 2 to the general case. \ Indeed, the function $f_\alpha$ may be written as
\bee
f_\alpha (t) = \frac {1}{\pi i} \int_{\R} F_\alpha (k) e^{-i k^2 t } dk 
\eee
where 
\bee
F_\alpha (k) &=& k \left [ G_0 (k) + g(k) G_r (k) \right ] 
\eee
and where $G_0 (k)$, $G_r (k)$ and $g(k)$ are respectively defined by (\ref {Equa15}), (\ref {Equa18}) and (\ref {Equa12}). \ Hence, by means of the Cauchy theorem 
and by making use of the same arguments given in the proofs of Lemma 4 and 5 it follows that 
\bee
f_\alpha (t) = \frac {1}{\pi i} \left [ \int_{\R} F_\alpha (k e^{-i \pi /4}) e^{- k^2 t } e^{-i \pi /4} dk - \frac {1}{2\pi i} \sum_{n=1}^\infty \beta_n  \mbox {Res} \ 
\left [ k  g(k) G_r (k) \right ]_{k=\sqrt {{\mathcal E}_{\alpha ,n}} } e^{-i {\mathcal E}_{\alpha ,n} t } \right ] 
\eee
where $\beta_n =1$ if the pole $k = \sqrt {{\mathcal E}_{\alpha ,n}}$ is inside the complex set enclosed by $\gamma$, and where $\beta_n =\frac 12$ 
if the pole $k = \sqrt {{\mathcal E}_{\alpha ,n}}$ belongs to the border $\gamma$, otherwise $\beta_n=0$. \ We should remark that we can apply the arguments by Lemma 4 and 5 provided that 
the functions $G_0 (k)$, $\ell_j (k)$ and $m_j (k)$, $j=1,2$, admits analytic extension in the sectors $\{ z \in \C \ :\ \mbox {arg} (z) \in [-\pi /4 ,0] \}$ and that such 
extensions may be suitably controlled by $e^{-i k^2 t}$ in the same domain. \ These conditions are both fulfilled provided that the test functions $\phi $ and $\psi$ are well 
localized. \ Furthermore, in the general case we have both families of complex poles (\ref {Equa8}) and (\ref {Equa9}). \ Finally, Watson's Lemma \cite {SFS} gives that 
\bee
&& \int_{\R} F_\alpha (k e^{-i \pi /4}) e^{- k^2 t } e^{-i \pi /4} dk = e^{-i \pi /4} \int_{0}^{+\infty } 
\left [ F_\alpha (k e^{-i \pi /4}) + F_\alpha (-k e^{-i \pi /4}) \right ]  e^{- k^2 t }  dk \\ 
&& \ \ = \frac 12 e^{-i \pi /4} \sum_{j=0}^N \frac {t^{-(j+1)/2}}{j!} \Gamma \left ( \frac {j+1}{2} \right ) 
\left . \frac {d^j \left [ F_\alpha (k e^{-i \pi /4}) + F_\alpha (-k e^{-i \pi /4}) \right ]}{d k^j} \right |_{k=0} + O \left ( t^{-(N+2)/2} \right ) \\ 
&& \ \ = e^{-i \pi /4} \sum_{s =0}^N \frac {t^{-(2s+1)/2}}{2s!} \Gamma \left ( \frac {2s+1}{2} \right ) 
\left . \frac {d^{2s}  F_\alpha (k e^{-i \pi /4}) }{d k^{2s}} \right |_{k=0} + O \left ( t^{-(N+2)/2} \right ) 
\eee
where a straightforward calculation gives that
\bee
F_\alpha (0) = 0 
\eee
for any $\alpha >0$. \ Theorem 1 is thus proved.
\begin {remark}
In fact, we could assume a weaker condition about the test functions $\phi$ and $\psi$; that is would be enough that 
\bee
|\phi (x) |, \, |\psi (x) | \le a e^{-bx^2} \, ,\ x \in \R 
\eee
for some $a,b >0$. 
\end {remark}

\appendix

\section {The Lambert function}

Here, we collect from \cite {K} some basic properties of the Lambert function. \ The Lambert function, denoted by $W(z)$,
 is defined to be the function satisfying the equation 
 \bee
 W(z)e^{W(z)} = z \, ,\ z \in \C \, .
 \eee%
 This function is a multivalued analytic function. \ The principal branch, denoted by $W_0 (z)$, is analytic at $z=0$ 
 and its power series expansion is given by
 \bee
 W_0 (z) = \sum_{n=1}^\infty \frac {(-n)^{n-1}}{n!}z^n \, . 
 \eee
 In particular $W_0(0)=0$. \ This series has radius of convergence $1/e$.
 
 In fact, if we denote $W_m (z)$, $m\in \Z$, all the branches of the Lambert function then we have the following picture:
 
 \begin {itemize}
  
  \item [-] $W_0 (z)$ is the only branch containing the interval $(-1/e , +\infty)$. \ It has a second-order 
  branch point at $z=-1/e$ and the branch cut is $(-\infty ,-1/e]$.
  
  \item [-] $W_1(z)$ and $W_{-1}(z)$ share the branch point at $z=-1/e$ with $W_0(z)$ and furthermore they 
  also have a branch point at $z=0$. \ Then, each of them has a double branch cut: $(-\infty -1/e ]$ and $(-\infty ,0]$. \ 
  By means of a suitable choice of the functions on the top of the branch cuts thus $W_0(z)$ and $W_{-1}(z)$ are 
  the only branches of the Lambert function that take real values.
  
  \item [-] All the other branches $W_m (z)$, $m\in \Z \setminus \{ 0,\pm 1\} $, have only the branch cut 
  $(-\infty ,0]$, similarly to the branches of the logarithm.
 
 \end {itemize}

 Among the properties of the Lambert function we recall the following ones:
 
 \begin {itemize}
  
 \item [-] {Symmetry conjugate property}: $W_m (\bar z) =\overline {W_{-m} (z)}$.
  
 \item [-] {Asymptotic expansion for large argument}: for large $z$ we have that 
 \be
 W_m (z) = \ln (z) + 2\pi i m - \ln \left [ \ln (z) + 2 \pi i m \right ] + \sum_{k=0}^{+\infty} 
 \sum_{j=1}^{+\infty} c_{kj} \frac {\ln^j \left [ \ln (z) + 2\pi i m \right ]}
 {\left [ \ln (z) + 2 \pi i m \right ]^{j+k}} \label {Equa24}
 \ee
 where
 \bee
 c_{kj} = \frac {(-1)^k}{j!} 
 \left [ 
 \begin {array}{c}
  k+j \\ k+1 
 \end {array}
 \right ]
 \eee
 where $
 \left [ 
 \begin {array}{c}
  k+j \\ k+1 
 \end {array}
 \right ]
 $ is a Stirling cycle number of first kind.
 
 \end {itemize}


\begin{thebibliography}{99}

\bibitem {Al} S. Albeverio, F. Gesztesy, R. Hoegh-Krohn, and H. Holden, 
{\it Solvable models in quantum mechanics}, (Berlin: Springer, 1988).

\bibitem {ChVi} F. Chevoir, and B. Vinter, {\it Scattering assisted tunneling in double barriers 
diode: scattering rates and valley current}, Phys. Rev. B {\bf 47} (1993) 7260–7274 (1993).

\bibitem {K} R.M. Corless, G.H. Gonnet, D.E. Hare, D.J. Jeffrey, and D.E. Knuth, {\it On the Lambert W function}, Adv. Comp. Math. {\bf 5} 329-359 (1996).

\bibitem {DSSW} D.C. Dobson, F. Santosa, S.P. Shipman, and M.I. Weinstein, {\it Resonances of a 
potential well with a thick barrier}, SIAM J. Appl. Math.
{\bf 73} 1489-1512 (2013).

\bibitem {FK} C. Ferrari, and H. Kovarik, {\it On the exponential decay of magnetic Stark resonances}, REp. Math. Phys. {\bf 56} 197-207 (2005).

\bibitem {G} Gottfried K., and T.M. Yan, {\it Quantum Mechanics: fundamentals}, Graduate texts in contemporary physics, $2^{\mbox {nd}}$ edition, (Springer: New York, 2003).

\bibitem {Ha} E.M. Harrell II, {\it Perturbation Theory and Atomic Resonances since Schr\"odinger's Time}, published as pp. 227-248 in: 
P. Deift, F. Gesztesy. P. Perry, and W. Schlag, eds., Spectral Theory and Mathematical Physics: A Festschrift in Honor of Barry Simon's 60th Birthday, 
Proceedings of Symposia in Pure Mathematics 76.1. Providence: American Mathematical Society, 2007.

\bibitem {He} I.W. Herbst, {\it Exponential Decay in the Stark Effect}, Commun. Math. Phys. {\bf 75} 197-205 (1980).

\bibitem {KS} H. Kovaric, and A. Sacchetti, {\it A nonlinear Schr\"odinger 
 equation with two symmetric point interactions in one dimension}, J. Phys. A: Math. Theor. {\bf 43} 155205 (2010).
 
\bibitem {SFS} Yu.V. Sidorov, M.V. Fedoriuk, and M.I. Shabunin, {\it Lectures on the theory of functions of a complex variable}, (MIR Publ.: Moscow 1985).

\bibitem {S} B. Simon, {\it Resonances in n-body quantum systems with dilatation analytic potentials and the foundations of time-dependent perturbation theory,}, Ann. Math. {\bf 97}, 247-274 (1973).

\bibitem {T} G. Teschl, {\it Mathematical Methods in Quantum Mechanics, with applications to Schr\"odinger operators}, Graduate Studies in Mathematics (AMS: 2009)

\bibitem {W} R. Weder, {\it $L^p - L^{p'}$ Estimates for the Schr\"odinger Equation on the Line and Inverse Scattering for the Nonlinear
Schr\"odinger Equation with a Potential}, J. of Funct. Anal. {\bf 170}, 37-68 (2000).

\bibitem {Zw} M. Zworski, {\it Resonances in Physics and Geometry}, Notices Amer. Math. Soc. {\bf 46} 319-328 (1999).

\end{thebibliography}
\end {document}